\documentclass[conference]{IEEEtran}
\IEEEoverridecommandlockouts
\usepackage{cite}
\usepackage{amsmath,amssymb,amsfonts}
\usepackage{bm}
\usepackage{algorithmic}
\usepackage{graphicx}
\usepackage{textcomp}
\usepackage{xcolor}
\usepackage{amsthm}
\usepackage{stfloats}
\usepackage{algorithm}
\usepackage[greek,english]{babel}
\usepackage{subcaption}
\usepackage{caption}
\def\BibTeX{{\rm B\kern-.05em{\sc i\kern-.025em b}\kern-.08em
    T\kern-.1667em\lower.7ex\hbox{E}\kern-.125emX}}
\newtheorem{mypro}{Proposition}

\makeatletter
\def\@cite#1#2{[{#1\if@tempswa , #2\fi}]}
\makeatother

\begin{document}

\title{Beamforming in Secure Integrated Sensing and \\
	Communication Systems with Antenna Allocation
}

\author{
	\IEEEauthorblockN{Yunxiang Shi\IEEEauthorrefmark{1}, Lixin Li\IEEEauthorrefmark{1}, Wensheng Lin\IEEEauthorrefmark{1}, Wei Liang\IEEEauthorrefmark{1}, and Zhu Han\IEEEauthorrefmark{2}}
	
	\IEEEauthorblockA{\IEEEauthorrefmark{1}School of Electronics and Information, Northwestern Polytechnical University,
		Xi'an, Shaanxi 710129, China}
	
	\IEEEauthorblockA{\IEEEauthorrefmark{2}University of Houston, Houston, USA}
	
	Email:\ shiyunxiang@mail.nwpu.edu.cn, \{lilixin, linwest, liangwei\}@nwpu.edu.cn, zhan2@uh.edu
	
	\thanks{This paper has been accepted for publication in IEEE Globecom 2024 workshop.}
	
}

\maketitle

\begin{abstract}
In this paper, we consider joint antenna allocation and transmit beamforming design in secure integrated sensing and communication (ISAC) systems. A dual-function base station (DFBS) aims to securely deliver messages to a single-antenna receiver while detecting potential eavesdroppers. To prevent eavesdropping, we incorporate specialized sensing signals, intentionally reducing communication signal power toward suspicious targets to improve sensing. We prioritize minimizing the matching error between the transmitting and required beampatterns for sensing and communication. Our design optimizes antenna allocation and beamforming at the DFBS, meeting minimum secrecy rate and power constraints. We propose solvers based on alternating optimization for the non-convex design problem. Simulations show that the antenna allocation scheme significantly improves safety performance.
\end{abstract}

\begin{IEEEkeywords}
Integrated sensing and communication, physical layer security,
transmit beamforming, antenna allocation.
\end{IEEEkeywords}

\section{Introduction}

Wireless emergency communications integrate sensing and localization to accurately detect and pinpoint the locations of individuals affected by natural disasters, enhancing the efficiency and effectiveness of rescue operations by providing real-time situational awareness and enabling targeted response efforts \cite{b1}. 5G and 6G technologies are gaining attention for revolutionizing communication with faster speeds and lower latency \cite{b2, Ma2024Latency}. Integrated sensing and communication (ISAC) systems, crucial for 5G and 6G networks, optimize spectrum use by combining sensing and communication in the same frequency \cite{b5,Wu2024ISAC}, reducing costs and enhancing performance. The spectrum supporting ISAC has a very wide coverage, including operation in the THz band. Moreover, ISAC design can be applied in multiple-input multiple-output (MIMO) systems, utilizing spatial diversity technology through beamformers to further improve performance \cite{b6}.

Due to the broadcast nature of wireless communication, the transmitted signals can be accessed by both legitimate users and eavesdroppers, thus the ISAC signals carrying the communication signals can be easily eavesdropped. Physical layer security (PLS) is an important innovative technology that can solve security issues in wireless communications from the physical level \cite{b7, Xiao2023Secure}.

The basic idea of PLS is to exploit the characteristics of the wireless channel, including noise, fading, interference, etc., so that the difference in performance between the legitimate receiver's link and the eavesdropper's link is significantly widened. In ISAC systems, physical layer security can also be used as a feasible solution. There has been some work applying physical layer security to ISAC systems \cite{b8}-\cite{b10}. In \cite{b8}, the authors consider a downlink secure ISAC system where a multi-antenna BS transmits classified information to a single-antenna UE while sensing a potential eavesdropper. In \cite{b9}, PLS in Reconfigurable Intelligent Surface (RIS)-assisted ISAC systems is studied, using Artificial Noise (AN) to interfere with eavesdroppers and maximize secrecy rates. In \cite{b10}, an AN-assisted secure downlink communication scheme for ISAC systems is explored, focusing on robust beamforming under imperfect CSI and probabilistic outages.

In addition, conventional digital beamforming schemes often require each element of the array antenna to be equipped with a separate radio frequency (RF) link, leading to higher hardware costs and design complexity \cite{b11}. The use of an antenna allocation strategy means that an optimization problem can be given greater design freedom without changing the number of antenna RF links, effectively replacing all-digital beamforming at a lower cost \cite{b12}.

Motivated by the above, this paper considers a secure ISAC system with DFBS equipped with a limited number of RF links, insufficient for designing each antenna. The DFBS handles communication and sensing tasks while preventing untrusted targets from stealing confidential information. We propose a sensing-centric design scheme to minimize the matching error between the transmitted and desired beampatterns under secrecy rate and RF link constraints. This is achieved by selecting suitable array antenna elements to transmit ISAC signals and designing the transmit beamformers. To solve the non-convex optimization problem, we propose an alternating optimization scheme based on one-dimensional search and penalized sequential convex programming. Simulation results demonstrate the effectiveness of the proposed scheme and the significant improvement brought by antenna allocation to the ISAC system. Our main contributions are summarized as follows:

\begin{itemize}
\item{We propose an alternative to digital beamforming in secure ISAC systems. Digital beamforming has high hardware cost and design complexity, so instead of equipping each antenna with a dedicated RF link, we choose a suitable number of antennas to match fewer RF links and transmit ISAC signals.}

\item{For the two different problems of security and antenna allocation, we propose an alternating optimization method, which can be solved efficiently by iterative solving. Experimental results show that our proposed scheme can allocate appropriate antenna elements to DFBSs according to the number of RF links under the precondition of security, and demonstrate that the use of antenna allocation in a secure ISAC system can effectively improve the perceptual and security performance of the ISAC system.}

\end{itemize}

\section{System Model}
We consider a downlink ISAC system as shown in Fig. 1, in which a DFBS communicates with a single antenna user and sense $T$ point targets in the far field of the DFBS. Among these targets, there are $T_u$ untrusted targets. The DFBS transmits both communication and sensing signals equipped with a uniform linear array (ULA) with $M>1$ antennas.

\subsection{Signal Model}
In the downlink, and without loss of generality, we assume that $ \mathcal T \triangleq \{ 1,...,T \}$ denotes the set of targets and $ \mathcal T_u \triangleq \{ 1,...,T_u \}$ denotes the set of untrusted targets. The DFBS jointly transmits a confidential communication stream $x_c[l]$, and $T$ sensing stream $\{x_t[l]\}_{t\in \mathcal{T}}$, where $l$ denoting the $l$'s symbol in communication/sensing streams. If $\bm{x}[l]\in\mathbb{C}^{M\times 1}$ denotes the transmit signal from the DFBS due to the $l$-th symbol, we can then write
\begin{equation}
	\label{eq1}
	\bm{x}[l]=\bm{w}_c x_c[l] + \sum_{t \in \mathcal{T}} \bm{w}_{s,t} x_t[l],
\end{equation}
where $\bm{w}_c\in\mathbb{C}^{M\times 1}$ denotes the communication beamforming vector, and $\bm{w}_{s,t}\in\mathbb{C}^{M\times 1}$ denotes the sensing beamforming vector for the stream $t \in \mathcal{T}$. The symbols are assumed to be independent random variable with zero mean and unit covariance. The beamforming vectors are subject to the total power constraint, $P_{all}$, given as
\begin{align}
	\label{eq2}
	\mathbb{E} [ \parallel \! \bm{x}[l] \! \parallel^{2} ] &= \parallel \! \bm{w}_c \! \parallel^{2} +  \sum_{t \in \mathcal{T}} \parallel \! \bm{w}_{s,t} \! \parallel^{2} \nonumber \\
	&= \text{Tr}(\bm{W}_c) +  \sum_{t \in \mathcal{T}} \text{Tr}(\bm{W}_{s,t}) = P_{all},
\end{align}
where $\bm{W}_c = \bm{w}_c \bm{w}_c^{H}$ and $\bm{W}_{s,t} = \bm{w}_{s,t} \bm{w}_{s,t}^{H}$. Note that $\bm{W}_{s,t}$ is a rank one matrix, i.e., $\text{rank}(\bm{W}_{st}) = 1$. Let $\bm{S} = \sum_{t \in \mathcal{T}} \bm{W}_{s,t}$, then $\bm{S}$ is a general rank, i.e., $0 \le T = \text{rank}(\bm{S}) \le M$.

\begin{figure}[!t]
	\centering{\includegraphics[width=3in]{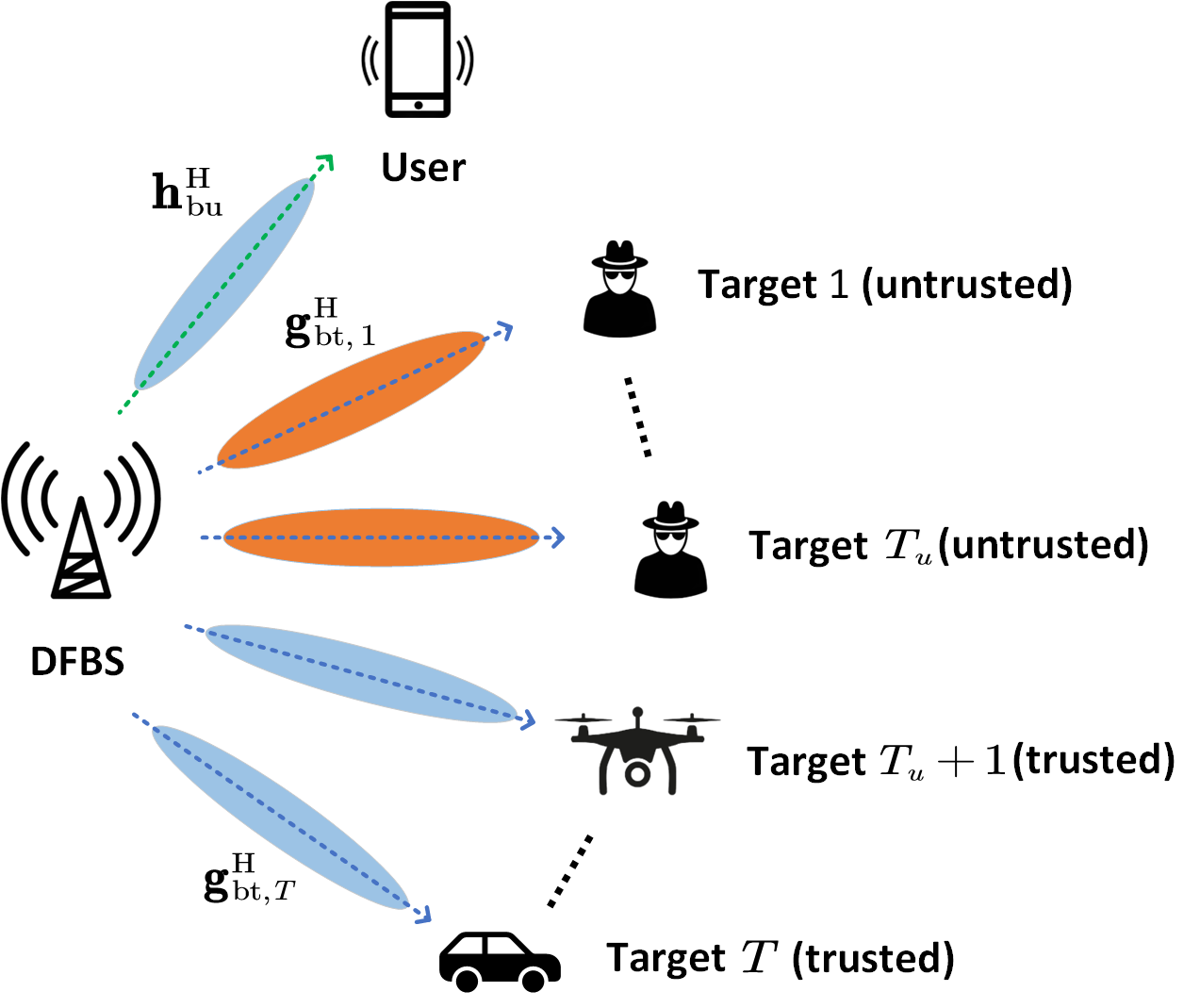}}
	\caption{A secure ISAC system with a ULA to perform communication and sensing functions.}
	\label{fig}
	\vspace{-1em}
\end{figure}

\subsection{Communication Model}
We denote the communication channel between the user and the DFBS as $\bm{h}_{bu} \in \mathbb{C}^{M\times 1}$. Next, considering a block fading channel model, where the channel remains constant over the transmission of the $L$ symbol, we can write the received signal at the user as
\begin{equation}
	\label{eq3}
	y_u[l]=\bm{h}_{bu}^H(\bm{w}_c x_c[l] + \sum_{t \in \mathcal{T}} \bm{w}_{s,t} x_t[l])+z,
\end{equation}
where $z\sim \mathcal C \mathcal N(0,\sigma^2_u)$ denotes the additive white Gaussian noise (AWGN) at the user receiver. Since sensing signals are equivalent to noise for communication signals, they can interfere with suspicious targets attempting to steal information. Accordingly, the user's communication signal-to-interference-plus-noise ratio (SINR) can be obtained as
\begin{align}
	\label{eq4}
	\text{SINR}_u &=\frac{{\mid\bm{h}_{bu}^H\bm{w}_c\mid}^2}{{\sum_{t \in \mathcal{T}} \mid\bm{h}_{bu}^H\bm{w}_{s,t}\mid}^2 + \sigma^2_u}
	= \frac{\bm{h}^H_{bu}\bm{W}_c\bm{h}_{bu}}{ \bm{h}^H_{bu}\bm{S}\bm{h}_{bu} + \sigma^2_u},
\end{align}

\subsection{Sensing Model}
For sensing model, we assume that both the DFBS and the targets are in low scattering areas and only consider the line of sight (LoS) channel between them. The angle of departure (AoD) from the BS to target $j \in \mathcal{T}$ is denoted by $\theta_j$, and the steering vector with angle $\theta_j$ is given as
\begin{equation}
	\label{eq5}
	\bm{a}(\theta_j)={[1,e^{j2\pi\frac{d}{\lambda}\text{sin}(\theta_j)},\dots,e^{j2\pi(M-1)\frac{d}{\lambda}\text{sin}(\theta_j)}]}^T,
\end{equation}
where $\lambda$ denotes the wavelength and $d$ denotes the spacing between two adjacent antennas. Then the DFBS-target channel $\bm{g}_{bt,j}$ for the $j$-th target is given by
\begin{equation}
	\label{eq6}
	\bm{g}^H_{bt,j} = \alpha_{bt,j}\bm{a}^H(\theta_j)\in\mathbb{C}^{1\times M},
\end{equation}
where $ \alpha_{bt,j}$ is the complex path gain. The received signal at the untrusted target $j$ is denoted as
\begin{equation}
	\label{eq7}
	y_{t,j}[l]=\bm{g}^H_{bt,j}\bm{w}_c x_c[l] + \sum_{t \in \mathcal{T}} \bm{g}^H_{bt,j}\bm{w}_{s,t} x_t[l]+z_j,
\end{equation}
where $z_j\sim \mathcal C \mathcal N(0,\sigma^2_j)$ denotes the AWGN at the $j$-th target receiver. Accordingly, the received SINR at the target $j$ is
\begin{equation}
	\label{eq8}
	\text{SINR}_j=\frac{{\mid\bm{g}^H_{bt,j}\bm{w}_c\mid}^2}{{\sum_{t \in \mathcal{T}} \mid\bm{g}^H_{bt,j}\bm{w}_{s,t}\mid}^2 + \sigma^2_j}
	= \frac{\bm{g}^H_{bt,j}\bm{W}_c\bm{g}_{bt,j}}{\bm{g}^H_{bt,j}\bm{S}\bm{g}_{bt,j} + \sigma^2_j}
\end{equation}

In this paper, we assume that the DFBS perfectly knows all the wireless channel state infromation (CSI). In practice, the channels $\bm{h}_{bu}$ and $\bm{g}_{bt,j}$ need to be estimated. We can use channel estimation techniques such as uplink pilots (respectively, the target echo signals) at the DFBS \cite{b13}.

\section{Problem Formulation}
The performance of the secure ISAC system can be considered from two aspects: communication and radar sensing metrics. With these metrics, the transmit beamformer and the antenna allocation scheme can be designed.

\subsection{Communication metrics}
According to \cite{b14}, it can be seen that the achievable secrecy rate of the multi-target ISAC communication system is determined by the difference between the user's achievable communication rate and the maximum achievable communication rate of all untrusted targets. In this case, the achievable secrecy rate at the user is given by
\begin{equation}
	\label{eq9}
	R_s=\mathop{\min}_{j\in{\mathcal T}_u} \ \log_2(1+\text{SINR}_u)-\log_2(1+\text{SINR}_j).
\end{equation}

\subsection{Radar sensing metrics}
Next, we use transmit power gain as the radar sensing performance metric. Radar systems detect and track targets by emitting signals in specific directions, achieved by designing multiple beams to match a desired beampattern.

In order to realize the detection function of the radar system and deliver information to the user, we need to ensure that targets and the user can obtain sufficient power. Let $\varphi_r$ denote the direction of the user from the DFBS, and $\epsilon$ denote the width degrees of rectangular beams with centers around $\{\theta_j\}^T_{j=1}$ and $\varphi_r$, the desired pattern is given by
\begin{equation}
	\label{eq10}
	D(\theta)=\left\{
	\begin{aligned}
		1 & , \ \ \text{if}\ \ \theta\in[\theta_j-\epsilon,\theta_j+\epsilon], \\
		1 & , \ \ \text{if}\ \ \theta\in[\varphi_r-\epsilon,\varphi_r+\epsilon], \\
		0 & , \ \ \text{elsewhere}.
	\end{aligned}
	\right.
\end{equation}
For any angle $\theta\in[-\frac{\pi}{2},\frac{\pi}{2}]$, the power radiated from the DFBS towards the direction $\theta$ is defined as the 
\begin{align}
	\label{eq11}
	J(\theta) &=\mathbb{E}[{\mid \bm{a}^H(\theta)(\bm{w}_c x_c[l] + \sum_{t \in \mathcal{T}} \bm{w}_{s,t} x_t[l])\mid}^2] \nonumber \\
	&=\bm{a}^H(\theta)(\bm{W}_c+\bm{S})\bm{a}(\theta),
\end{align}

An important metric to measure radar performance is beampattern matching error, which measures the difference between the desired beampattern and the actual beampattern. For a desired beampattern $D(\theta)$, the beampattern mismatch error evaluated over the $Q$ discrete sample angles $\{\tilde{\theta}_q\}^{Q}_{q=1}$ is
\begin{equation}
	\label{eq12}
	K=\frac{1}{Q}\mathop{\Sigma}_{q=1}^{Q}{\mid J(\tilde{\theta}_q)-\mu D(\tilde{\theta}_q)\mid}^{2},
\end{equation}
where $\mu$ is the unknown autoscale parameter. By increasing the value of $Q$, higher computational complexity can be chosed to get more accurate beampattern matching.

\subsection{Design problems}
We now consider the problems of designing the transmit beamformers and the antenna allocation scheme at the DFBS. Antenna allocation strategies can be an effective alternative to all-digital beamforming. 

To allocate the antennas at the DFBS, assuming that there are $G$ RF transmission links and $G<M$, the antenna allocation strategy refers to selecting the optimal $G$ antennas from $M$ antennas under the safe rate and transmit power constraints and designing the corresponding beamforming vectors so as to minimize the beampattern matching error. Define the binary vector $\bm{u} = [u_1, \dots, u_M]$. The entry $u_m$ is the antenna allocation coefficient for the $m$-th antenna. Hence, $u_m = 1$ if the $m$-th antenna is selected, and is zero otherwise. Accordingly, we formulate the design problem as
\begin{subequations}\label{eq13}
	\begin{align}		
		(\text{P1}):\: \mathop{\min}_{\bm{w}_c,\bm{w}_{s,t},\mu, \bm{u}} \ & \frac{1}{Q}\mathop{\Sigma}_{q=1}^{Q}{\mid J(\tilde{\theta}_q)-\mu D(\tilde{\theta}_q)\mid}^{2} \nonumber \\
		s.t. \ \ & R_s \ge R_0, \label{eq13a} \\
		& \parallel \! \bm{w}_c \! \parallel^{2} +  \sum_{t \in \mathcal{T}} \parallel \! \bm{w}_{s,t} \! \parallel^{2}=P_{all}, \label{eq13b} \\
		& | w_c(m) |^{2} +  \sum_{t \in \mathcal{T}} | w_{s,t}(m) |^{2}\le u_m P_b, \forall m, \label{eq13c} \\
		& u_m \in \{0,1\}, \label{eq13d} \\
		& \sum_{m = 1}^{M} u_m = G. \label{eq13e}
	\end{align}
\end{subequations}
Here, $R_0$ is the minimum secrecy rate constraint with (\ref{eq13b}), and $P_{all}$ is the total transmit power with (\ref{eq13c}). The $P_b$ is the maximum allowable power for each antenna at the DFBS. (\ref{eq13d}) indicates that each antenna in the ULA of the DFBS has two states, selected or unselected, and (\ref{eq13e}) limits the maximum number of antennas that can be selected, which is equal to the total number of available RF links.

\section{Proposed Method}

In this section, we develop solvers for (P1) based on alternating optimization, the details of the proposed methods are summarized in Algorithm 1.

Considering the problem is a non-convex because of the non-convex objective function, the non-convex secrecy rate in (\ref{eq13c}), and the binary variable $\{u_m\}_{m=1}^{M}$. In the following, we use the a penalized sequential convex programming (P-SCP) \cite{b15} to address the $\{u_m\}_{m=1}^{M}$. Specifically, the binary constraint (\ref{eq13d}) can be rewritten in equivalent form as:
\begin{subequations}\label{eq14}
	\begin{align}		
		& 0 \le u_m \le 1, m \in \mathcal{M}, \label{eq14a} \\
		& \sum_{m = 1}^{M} (u_m - u^2_m) \le 0. \label{eq14b}
	\end{align}
\end{subequations}
where $\mathcal{M} = \{1, \dots, M \} $ denotes the set of antenna indices, constraint (\ref{eq14b}) is the difference of a linear and a
quadratic term. Notice that $\{u_m\}_{m=1}^{M}$ is infinitely close to 0 or 1 when the left term of (\ref{eq14b}) is infinitely close to 0, so we can move the restriction (\ref{eq14b}) into the objective function and multiply it by a sufficiently large penalty factor $\eta$.

\begin{algorithm}[!h]
	\caption{Alternating Optimization Method for Secure ISAC Systems with Limited RF Links.}
	\label{alg1}
	\begin{algorithmic}[1]
		\renewcommand{\algorithmicrequire}{\textbf{Input:}}
		\renewcommand{\algorithmicensure}{\textbf{Output:}}
		
		\REQUIRE $\{ \theta_j \}^{T}_{j=1}$, $\varphi_r$, $P_{all}$,  $P_b$, $R_0$, $\eta$, and $\epsilon$.
		\ENSURE $\bm{w}_c$, $\bm{S}$, and $\bm{u}$.
		\STATE Initialize $\hat{\bm{u}}^{(0)}$ and set $k = 0$.
		\REPEAT
		\STATE Set $k = k + 1$.
		\REPEAT
		\STATE Solve Problem (SDR P5) to obtain ${\bm{W}^{(k)}_c}^{\ast}$, ${\bm{S}^{(k)}}^{\ast}$, ${\mu^{(k)}}^{\ast}$ , and ${\bm{u}^{(k)}}^{\ast}$ with $\hat{\bm{u}}^{(k-1)}$
		
		\STATE Obtain $\hat{\bm{w}}^{(k)}_c$, $\hat{\bm{S}}^{(k)}$, $\hat{\mu}^{(k)}$ and $\hat{\bm{u}}^{(k)}$ according to (\ref{eq24a}), (\ref{eq24b}), (\ref{eq24c}), and (\ref{eq24d}).
		
		\UNTIL{find $\lambda_G$ such that $f(\lambda_G)$ is the smallest according to (\ref{eq22})}.
		\UNTIL{convengence.}
		\RETURN $\bm{w}_c = \hat{\bm{w}}^{(k)}_c$, $\bm{S} = \hat{\bm{S}}^{(k)}$, and $\bm{u} = \hat{\bm{u}}^{(k)}$.
	\end{algorithmic}
\end{algorithm}

Then, problem (P1) can be reformulated as
\begin{subequations}\label{eq15}
	\begin{align}		
		(\text{P2}):\: \mathop{\min}_{\bm{w}_c,\bm{w}_{s,t},\mu,\bm{u}} \ & \frac{1}{Q}\mathop{\Sigma}_{q=1}^{Q}{\mid J(\tilde{\theta}_q)-\mu D(\tilde{\theta}_q)\mid}^{2} + \eta H(\bm{u})\nonumber \\
		s.t. \ \ & (\ref{eq13a}),(\ref{eq13b}),(\ref{eq13c}),(\ref{eq13e}), (\ref{eq14a}), \label{eq15a}
	\end{align}
\end{subequations}
where
\begin{equation}
	\label{eq16}
	H(\bm{u}) = \sum_{m = 1}^{M} (u_m - u^2_m).
\end{equation}
Notice that $\sum_{m = 1}^{M} u^2_m$ is a differentiable convex function, according to the subgradient inequality of convex functions, we have
\begin{equation}
	\label{eq17}
	\sum_{m = 1}^{M} u^2_m \ge \sum_{m = 1}^{M} {u^{(k)}_m}^2 + \sum_{m = 1}^{M} 2 u^{(k)}_m(u_m - u^{(k)}_m)
\end{equation}
where $u^{(k)}_m$ denotes the value of $u_m$ in the $l$-th iteration. We relax the quadratic term of $H(\bm{u})$ in the objective function of problem (P2) by replacing it with its lower bound, and the resulting new problem (P3) will minimize the upper bound of the objective function of problem (P2). We define
\begin{equation}
	\label{eq18}
	\tilde{H}(\bm{u}) = \sum_{m=1}^{M} u_m - {u^{(k)}_m}^2 - 2 u^{(k)}_m(u_m - u^{(k)}_m)
\end{equation}
Then problem (P2) can be solved iteratively by the following optimization problem (P3)
\begin{subequations}\label{eq19}
	\begin{align}		
		(\text{P3}):\: \mathop{\min}_{\bm{w}_c,\bm{w}_{s,t},\mu, \bm{u}} \ & \frac{1}{Q}\mathop{\Sigma}_{q=1}^{Q}{\mid J(\tilde{\theta}_q)-\mu D(\tilde{\theta}_q)\mid}^{2} + \eta \tilde{H}(\bm{u}) \nonumber \\
		s.t. \ \ & (\ref{eq13a}),(\ref{eq13b}),(\ref{eq13c}),(\ref{eq13e}), (\ref{eq14a}), \label{eq19a}
	\end{align}
\end{subequations}

Next, we deal with the constraint (\ref{eq13a}). Consider the non-convex constraint, accoriding to \cite{b16}, auxiliary optimizations $\lambda_G$ and $\beta$ are introduced. $\lambda_G$ denotes the maximum SINR among all untrusted targets, and $\beta=2^{R_0}(1+\lambda_G)-1$. Therefore, the problem (P3) can be reformulated as
\begin{subequations}\label{eq20}
	\begin{align}
		(\text{P4}):\: \mathop{\min}_{\bm{w}_c,\bm{w}_{s,t},\mu,\bm{u},\lambda_G} \ & \frac{1}{Q}\mathop{\Sigma}_{q=1}^{Q}{\mid J(\tilde{\theta}_q)-\mu D(\tilde{\theta}_q)\mid}^{2} + \eta \tilde{H}(\bm{u}) \nonumber \\
		s.t. \ \ & \bm{g}^H_{bt,j}\bm{W}_c\bm{g}_{bt,j} \le \lambda_G (\bm{g}^H_{bt,j}\bm{S}\bm{g}_{bt,j} + \sigma^2_j), \nonumber \\
		& \ \ \ \ \ \ \ \ \ \ \ \ \ \ \ \ \ \ \ \ \ \ \ \ \ \ \ \forall j \in {\mathcal T}_u, \label{eq20a} \\
		& \bm{h}^H_c \bm{W}_c \bm{h}_c \ge \beta (\bm{h}^H_c \bm{S} \bm{h}_c + \sigma^2_u), \label{eq20b} \\
		& (13b),(13c),(13e),(14a). \label{eq20c}
	\end{align}
\end{subequations}
Notice that the constraints in (\ref{eq20a}) and (\ref{eq20b}) are obtained by splitting the constraint in (\ref{eq13a}). For any given $\lambda_G > 0$, the optimization of $\bm{w}_c,\bm{w}_{s,t},\mu$, and $\bm{u}$ in problem (P4) becomes
\begin{subequations}\label{eq21}
	\begin{align}
		(\text{P5}):\: \mathop{\min}_{\bm{w}_c,\bm{w}_{s,t},\mu,\bm{u}} \ & \frac{1}{Q}\mathop{\Sigma}_{q=1}^{Q}{\mid J(\tilde{\theta}_q)-\mu D(\tilde{\theta}_q)\mid}^{2} + \eta \tilde{H}(\bm{u}) \nonumber \\
		s.t. \ \ & (13b),(13c),(13e),(14a),(20a),(20b). \label{eq21a}
	\end{align}
\end{subequations}

For the problem (P3), We can solve it in two steps. Let $f(\lambda_G)$ denote the optimal objective value achieved by problem (P5) with given $\lambda_G$. In the first step, we solve problem (P5) with any given $\lambda_G > 0$. Then, 1D search method is used for $\lambda_G$ in following problem (P3.1).
\begin{equation}
	\label{eq22}
	(\text{P3.1}):\: \mathop{\min}_{\lambda_G > 0} \ f(\lambda_G)
\end{equation}

However, problem (P5) is still a non-convex  quadratically constrained quadratic program (QCQP) problem. To resolve this issue, we use the technique of semidefinite relaxation (SDR) to relax the rank-one constraint and accordingly obtain the SDR version of problem (SDR P5) as
\begin{subequations}\label{eq23}
	\begin{align}
		(\text{SDR P5}):\: \mathop{\min}_{\bm{W}_c,\bm{S},\mu,\bm{u}} \ & \frac{1}{Q}\mathop{\Sigma}_{q=1}^{Q}{\mid J(\tilde{\theta}_q)-\mu D(\tilde{\theta}_q)\mid}^{2} + \eta \tilde{H}(\bm{u}) \nonumber \\
		s.t. \ \ & \text{Tr}(\bm{W}_c) + \text{Tr}(\bm{S}) = P_{all}, \label{eq23a} \\
		& \bm{W}_c(m,m) + \bm{S}(m,m) \le u_m P_b, \label{eq23b} \\
		& \bm{W}_c \succeq 0, \bm{S} \succeq 0, \label{eq23c} \\
		& (13e), (14a), (20a), (20b).
	\end{align}
\end{subequations}
It is observed that problem (SDR P5) is a convex quadratic semidefinite programing (QSDP) problem that can be solved optimally by convex solvers such as CVX \cite{b17}. Let $\bm{W}_c^{\ast}$, $\bm{S}^{\ast}$, $\mu^{\ast}$ and $\bm{u}^{\ast}$ denote the obtained optimal solution to problem (SDR P5), where $\bm{W}_c^{\ast}$ is generally of high rank. Based on $\bm{W}_c^{\ast}$, we can reconstruct an equivalent rank-one solution and accordingly find the optimal solution to problem (P5), as shown in the following proposition.

\begin{mypro}
	The optimal solution to problem (P5) is 
\begin{subequations}
	\label{eq24}
	\begin{align}
		& \hat{\bm{W}_c}=\hat{\bm{w}}_c\hat{\bm{w}}^H_c, \label{eq24a} \\
		& \hat{\bm{S}}=\bm{S}^{\ast}+\bm{W}^{\ast}_c-\hat{\bm{W}}_c, \label{eq24b} \\
		& \hat{\mu} = \mu^{\ast}, \label{eq24c} \\
		& \hat{\bm{u}} = \bm{u}^{\ast}, \label{eq24d}
	\end{align}
\end{subequations}
where $\hat{\bm{w}}_c=(\bm{h}^H_{bu} \bm{W}^{\ast}_c\bm{h}_{bu})^{-\frac{1}{2}} \bm{W}^{\ast}_c \bm{h}_{bu}$ denotes the corresponding transmit beamforming vector at the DFBS. Accordingly, $\hat{\bm{W}}_c$, $\hat{\bm{S}}$, $\hat{\mu}$ ,and $\hat{\bm{u}}$ are the optimal solution to problem (P5). 
\end{mypro}

\begin{proof}
	See Appendix A.
\end{proof}

\textit{Remark 1:} Notice that $\bm{S} = \sum_{t \in \mathcal{T}} \bm{W}_{s,t}$ and $\text{rank}(\bm{S}) = T$. This corresponds to the case where there are $T$ sensing beams, each of which can be obtained by an eigenvalue decomposition (EVD) of $\bm{S}$, which in turn yields the corresponding beamforming vectors $\{\bm{w}_{s,t}\}^{T}_{t=1}$ of the sensing signal.

\section{Simulation Results}

\begin{figure} [t!]
	\centering
	\subfloat[With antenna allocation.]{\label{fig2:a}
		\includegraphics[width=2.7in]{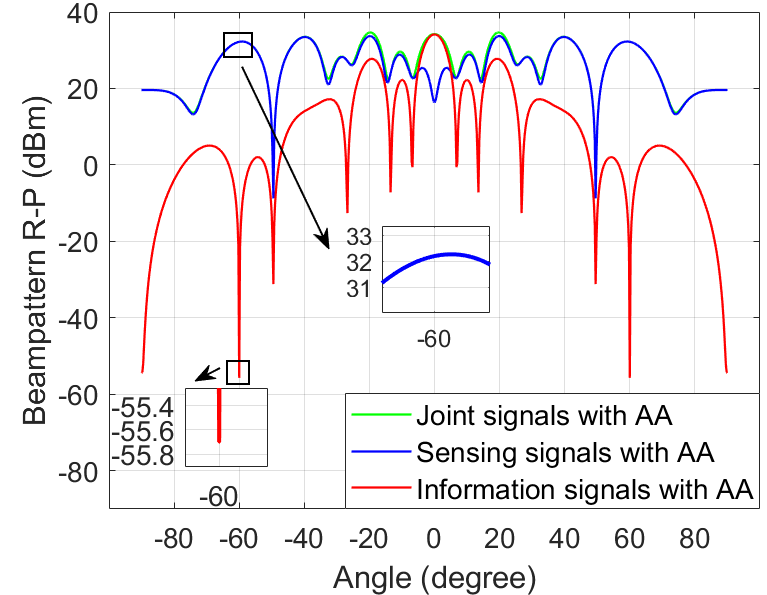}}
	\\
	\subfloat[Without antenna allocation.]{\label{fig2:b}
		\includegraphics[width=2.7in]{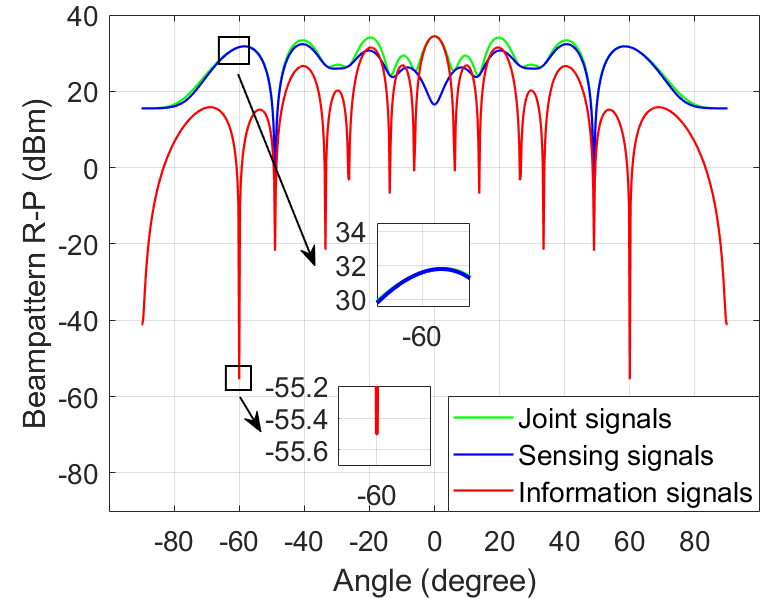}}
	\caption{Beampattern radiated power under different designs, where $R_0 = 5$bps/Hz, $P_{all} = 30$dBm.}
	\label{fig2}
	\vspace{-1em}
\end{figure}

This section illustrates the benefits of introducing an antenna allocation scheme in a secure ISAC system through numerical experiments. In the simulation, the DFBS is modeled as a ULA with $M=16$ half-wavelength spaced elements, and there are $T=6$ targets located at angles $-60^{\circ}$, $60^{\circ}$, $-40^{\circ}$, $40^{\circ}$, $-20^{\circ}$, and $20^{\circ}$. The targets located at $-60^{\circ}$ and $60^{\circ}$ are assumed to be untrusted eavesdroppers. The user is located at $0^{\circ}$.

We model both the communication channel $\bm{h}^{H}_{bu}$ and the sensing channel $\{ \bm{g}^{H}_{bt,j} \}_{j \in \mathcal{T}}$ as line-of-sight channels. The user and all eavesdroppers have the same receiver noise power of $-60\text{dBm}$. For an ideal beampattern $D(\tilde{\theta}_l)$, the width of rectangular box is set as $\epsilon = 5^{\circ}$. For simplicity, in the following figures, we denote antenna allocation as AA, user radiated power as R-P, and target illumination power as I-P.

\begin{figure} [t!]
	\centering
	\includegraphics[width=2.7in]{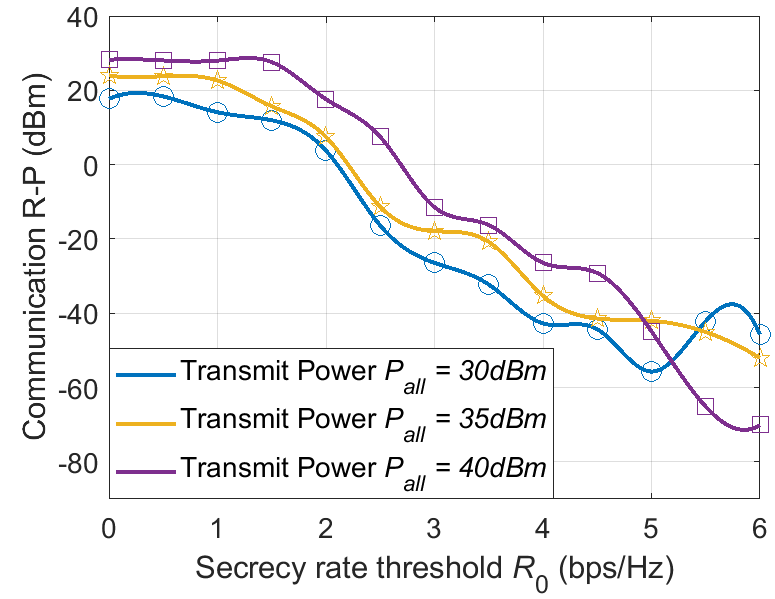}
	\caption{Communication signals radiated power directed to untrusted targets versus the secrecy rate threshold $R_0$.}
	\label{fig3}
\end{figure}

\begin{figure} [t!]
\vspace{-1em}
	\centering
	\includegraphics[width=2.7in]{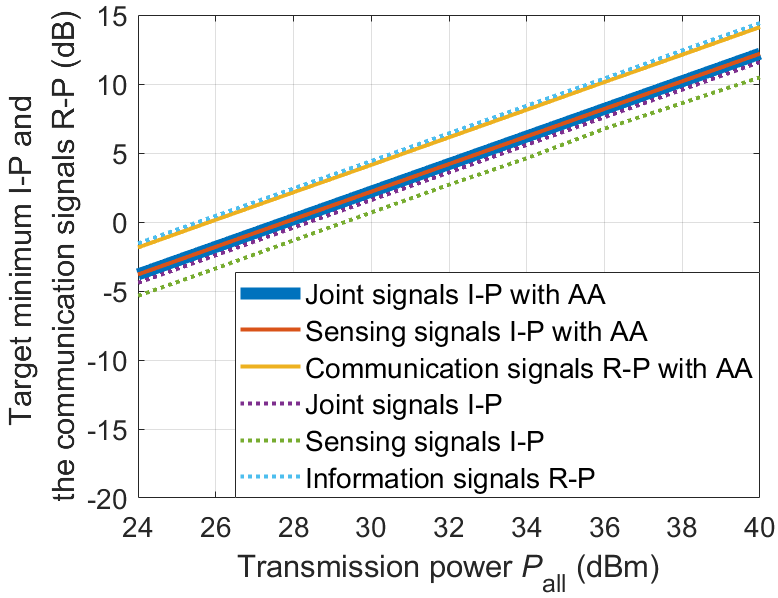}
	\caption{Target minimum radiated power and the communication signals radiated power with $R_0 = 5$bps/Hz. }
	\label{fig4}
	\vspace{-1em}
\end{figure}

Fig. 2 gives the beam direction map gain with and without the antenna assignment scheme in both cases. From Fig. 2(a), it can be seen that the peak gain of the beam direction map of the sensed signal with the antenna assignment scheme is higher and the gain at the same angle is higher compared to Fig. 2(b). In terms of security performance, Fig. 2(a) shows that the antenna allocation scheme makes the antenna beam direction map gain lower at the suspicious target (eavesdropper), which increases the difficulty for the eavesdropper to obtain the communication information.

Fig. 3 gives the communication power at the suspected target under different minimum secrecy rate constraints. Based on the symmetry design, the communication power of the two eavesdroppers remains consistent. It can be seen that the power of the information signals towards the eavesdroppers for different transmit power scenarios shows a decreasing trend as the secrecy rate threshold increases. It is worth noting that for transmit power $P_{all} = 30$, the communication power instead rises after the secrecy rate threshold $R_0 > 5$, which is due to the fact that the current transmit power is not sufficient to provide such a high secrecy rate.

In Fig. 4, we give the minimum illumination power pointing in the direction of each target as well as the variation of the communication signal power pointing towards the user for different transmit power $P_{all}$. In this case, both the joint and sensing signal transmit power with the antenna allocation scheme is better than the case without the antenna allocation scheme. In the user direction, the communication signal transmit power with the antenna allocation scheme is slightly lower than the result without the antenna allocation scheme. The above results show that the antenna allocation scheme adjusts the power allocation structure of the DFBS and improves the safety performance of the secure ISAC system, which in turn proves the effectiveness of the scheme.

\section{Conclusion}
In this paper, we consider joint beamforming and antenna assignment in a secure ISAC system. Selection of appropriate antennas for the RF link of DFBS can effectively improve the safety performance in ISAC systems. By slightly reducing power toward the communication user, the illumination power of other targets including the suspicious targets is improved and the power of the information signal pointing in the direction of the suspicious targets is further reduced, which further improves the security of the ISAC system with enhanced perception performance, which proves the effectiveness of our proposed scheme.

\section*{Appendix A \\ Proof of Proposition 1}
From (\ref{eq24a}), we can see that $\hat{\bm{S}} + \hat{\bm{W}}_c = \bm{S}^{\ast}+\bm{W}^{\ast}_c $, and hence $\hat{\bm{S}}$ and $\hat{\bm{W}}_c$ satisfy the constraint (\ref{eq23a}). According to (\ref{eq24d}), it is clear that $\hat{\bm{S}}$, $\hat{\bm{W}}_c$, and $\hat{\bm{u}}$ satisfy the constraint (\ref{eq23b}).

Notice that $\bm{W}^{\ast}_c \succeq 0$, we have $\bm{W}^{\ast}_c = \tilde{\bm{W}} \tilde{\bm{W}}^{H}$. Based on this, for any $\bm{z} \in \mathbb{C}^{M \times 1}$, it follows that
\begin{align}
	\label{eq25}
	&\bm{z}^H (\bm{W}^{\ast}_c - \hat{\bm{W}}_c) \bm{z} = \bm{z}^H \bm{W}^{\ast}_c \bm{z} - \bm{z}^H \frac{\bm{W}^{\ast}_c \bm{h}_{bu} \bm{h}^{H}_{bu} \bm{W}^{\ast}_c}{\bm{h}^H_{bu} \bm{W}^{\ast}_c\bm{h}_{bu}} \bm{z} \nonumber \\
	&= \frac{1}{\bm{h}^H_{bu} \bm{W}^{\ast}_c\bm{h}_{bu}} (\bm{z}^H \bm{W}^{\ast}_c \bm{z} \bm{h}^H_{bu} \bm{W}^{\ast}_c\bm{h}_{bu} - \bm{z}^{H} \bm{W}^{\ast}_c \bm{h}_{bu} \bm{h}^{H}_{bu} \bm{W}^{\ast}_c \bm{z} ) \nonumber \\
	& = \frac{1}{\bm{h}^H_{bu} \bm{W}^{\ast}_c\bm{h}_{bu}} (\parallel \! \bm{a} \! \parallel^2 \parallel \! \bm{b} \! \parallel^2 - | \bm{a}^H \bm{b} |^2) \ge 0,
\end{align}
where $\bm{a} = \tilde{\bm{W}}^{H} \bm{z} \in \mathbb{C}^{M \times 1}$, $\bm{b} = \tilde{\bm{W}}^{H} \bm{h}_{bu} \in \mathbb{C}^{M \times 1}$, and the condition for the inequality to hold in the last step of the (\ref{eq25}) is the Cauchy-Schwarz inequality. Therefore, $\bm{W}^{\ast}_c - \hat{\bm{W}}_c \succeq 0$ holds. By using this together with (24b), we have $\hat{\bm{S}} \succeq \bm{S}^{\ast} \succeq 0$. So the constraint (\ref{eq23c}) holds for $\hat{\bm{S}}$ and $\hat{\bm{W}_c}$.

It is clear that from (\ref{eq24a}) and (\ref{eq24b}) that $\bm{h}^{H}_{bu} \bm{W}^{\ast}_c \bm{h}_{bu} = \bm{h}^{H}_{bu} \hat{\bm{W}}_c \bm{h}_{bu}$ and $\bm{h}^{H}_{bu} \bm{S}^{\ast} \bm{h}_{bu} = \bm{h}^{H}_{bu} \hat{\bm{S}} \bm{h}_{bu}$, and therefore, constraint (\ref{eq20b}) holds. Considering that $\bm{W}^{\ast}_c - \hat{\bm{W}}_c \succeq 0$ and $\hat{\bm{S}} \succeq \bm{S}^{\ast} \succeq 0$, we have
\begin{align}
	\label{eq26}
	& \bm{g}^{H}_{bt,j} \hat{\bm{W}}_c \bm{g}_{bt,j} \le \bm{g}^{H}_{bt,j} \bm{W}^{\ast}_c \bm{g}_{bt,j} \le \lambda_G(\bm{g}^{H}_{bt,j} \bm{S}^{\ast} \bm{g}_{bt,j} + \sigma^2_j) \nonumber \\
	& \le \lambda_G(\bm{g}^{H}_{bt,j} \hat{\bm{S}} \bm{g}_{bt,j} + \sigma^2_j), \forall j \in \mathcal{T}_u.
\end{align}
As a result, constraint (\ref{eq20a}) is satisfied.

Finally, consider the objective function, from the expression of the objective function, it is clear that $\hat{\bm{W}}_c$, $\hat{\bm{S}}$, $\hat{\mu}$ and $\hat{\bm{u}}$ achieve the same objective value for problem (SDR P5) as that achieved by $\bm{W}^{\ast}_c$, $\bm{S}^{\ast}$, $\mu^{\ast}$ and $\bm{u}^{\ast}$.

From the above proof it can be concluded that $\hat{\bm{W}}_c$, $\hat{\bm{S}}$, $\hat{\mu}$ and $\hat{\bm{u}}$ are optimal solutions to the problem (SDR P5). Notice that $\text{rank}(\hat{\bm{W}}_c) = 1$ with $\hat{\bm{W}}_c = \hat{\bm{w}}_c \hat{\bm{w}}^{H}_c$. Therefore, $\hat{\bm{W}}_c$, $\hat{\bm{S}}$, $\hat{\mu}$ and $\hat{\bm{u}}$ are also optimal for problem (P5). Thus completing the proof of proposition 1.

\end{document}